\documentclass{article}

\usepackage[numbers]{natbib}
 \usepackage{graphicx,overpic}
\usepackage{amsmath}

\usepackage{amssymb}
\usepackage{amsthm}

\usepackage{url}
\usepackage{listings} 

\usepackage{natbib}
\usepackage{color}
\newtheorem{theorem}{Theorem}
\newtheorem{lemma}[theorem]{Lemma}
\newtheorem{Lemma*}[theorem]{Lemma}
\newtheorem{remark}[theorem]{Remark}

\newtheorem{definition}[theorem]{Definition}

\newtheorem{example}[theorem]{Example}

\def\jdlqed{\vbox{\hrule \hbox{\vrule\hbox to
5pt{\vbox to 6pt{\vfil}\hfil}\vrule}\hrule}}

\newcommand{\R}{\mathbb R}

\def\comment#1{\textit{[#1]}}
\def\comment#1{}

\begin{document}

\title{Tropical neural networks and its applications to classifying phylogenetic trees}
\author{Ruriko Yoshida \and Georgios Aliatimis \and Keiji Miura}

\maketitle

\begin{abstract}
Deep neural networks show great success when input vectors are in an Euclidean space. However, those classical neural networks show a poor performance when inputs are phylogenetic trees, which can be written as vectors in the tropical projective torus. Here we propose tropical embedding to transform a vector in the tropical projective torus to a vector in the Euclidean space via the tropical metric. We introduce a tropical neural network where the first layer is a tropical embedding layer and the following layers are the same as the classical ones. We prove that this neural network with the tropical metric is a universal approximator and we derive a backpropagation rule for deep neural networks. 
Then we provide TensorFlow 2 codes for implementing a tropical neural network in the same fashion as the classical one, where the weights initialization problem is considered according to the extreme value statistics.
We apply our method to empirical data including sequences of hemagglutinin for influenza virus from New York.
Finally we show that a tropical neural network can be interpreted as a generalization of a tropical logistic regression.
\end{abstract}

\section{Introduction}
A neural network is a learning method, called deep learning, to learn data in a way to mimic a brain system, i.e., which interconnects nodes, called neurons, in a layered structure like a human brain system \cite{Goodfellow16, Calin20, Grohs23}.
Recent years deep neural networks show great success to process input data which lay in the Euclidean space \cite{Ford18}.
However, when input data are phylogenetic trees or the time series with trends, represented as vectors in the {\em tropical projective torus} \cite{Nye, NTWY, MLKY, KDETree, YZZ, Trop_KDE, 10.1093/bioinformatics/btaa564, YTMM, tang2020tropical}, classical neural networks show a poor performance. Therefore in this paper, we propose neural networks which process an input data as vectors over the tropical projective torus.  The tropical projective torus denoted by $\mathbb{R}^d/\mathbb{R}\mathbf{1}$ is the $d$-dimensional  real numbers, $\mathbb{R}^d$, mod by the vector with all ones, i.e., $\mathbf{1}:= (1, 1, \ldots, 1) \in \mathbb{R}^d$.  Over the tropical projective torus denoted by $\mathbb{R}^d/\mathbb{R}\mathbf{1}$, we define $x:=(x_1, x_2, \ldots, x_d) = (x_1+c, x_2+c, \ldots , x_d+c) \in \mathbb{R}^d/\mathbb{R}\mathbf{1}$ where $c \in \mathbb{R}$ \cite{MS}.  
Here we consider the {\em tropical metric}, also known as the {\em generalized Hilbert projective metric}, over the tropical projective torus as 
activation functions in a hidden layer of a neural network.
It is important to keep the invariance of the input vector under the one-vector which is innate in the tropical projective torus \cite{MS, joswigBook, SS, LY, Trop_HAR, firststep}. Our strategy is to embed an input vector in the tropical projective torus into a vector in the classical Euclidean space in the first layer. This is analogous to the word embedding in the field of natural language processing 
\cite{Vaswani17, onan2021sentiment}. Then the following layers can be the same as the classical ones. 

Although some previous works analyzed ReLU neural networks by using the tropical geometry,
the neural networks themselves are defined on a classical Euclidean space \cite{Zhang18, Alfarra22, Montúfar22}.
In this paper, on the other hand, we consider a tropical projective torus as an input space and keep the invariance under the one-vector. That is, our work is truly tropical.

In this paper, we first introduce a tropical embedding layer. We use the tropical embedding layer as the first layer of the classical neural networks to keep the invariance in the tropical projective space.
To check if this tropical neural network has enough flexibility, we next prove that this tropical neural network is a universal approximator.
Then we derive a backpropagation rule for the tropical neural networks.
We provide TensorFlow 2 codes for implementing a tropical neural network in the same fashion as the classical one, where the weights initialization problem is considered according to the extreme value statistics.
We show its applications to phylogenomics, a new field in evolutionary biology which applies tools from phylogenetic trees to genome data.
Applications includes simulated data under the multi-species coalescent model which is the most popular model to analyze gene tree analysis on a genome data \cite{mesquite}, and empirical data of influenza virus data set collected from the state of New York \cite{NTWY}.
Finally we briefly show that a tropical neural network can be interpreted as a generalization of a tropical logistic regression.

\section{Tropical Embedding for Tropical Neural Networks}
The classical neural networks only accept a input vector in an Euclidean space in its original form.
Thus they cannot accept a phylogenetic tree as an input since a space of phylogenetic trees is not Euclidean \cite{SS,AK,BHV}, for example.
Therefore, we first consider a tropical embedding layer, which is analogous to the word embedding in natural language processing \cite{Vaswani17, onan2021sentiment}.
Once a phylogenetic tree is embedded in an Euclidean space, a classical neural network can is applied to analyzing it.

\begin{definition}[tropical neural networks]
A tropical neural network is a network where a tropical embedding layer as the first hidden layer is followed by a classical neural network (classical layers).
\end{definition}

\begin{definition}[tropical embedding layer]
Let $x$ in $\mathbb{R}^d/\mathbb{R}\mathbf{1}$ be an input vector to the tropical embedding layer.
the activity of $j$-th neuron as an output of the tropical embedding layer is given by
\begin{equation}
z_j = \max_i(x_i + w^{(1)}_{ji}) - \min_i(x_i + w^{(1)}_{ji}).
\label{eq_activation_function}
\end{equation}
\end{definition}

\begin{remark}
Note that no activation function is executed for $z$ as ``max - min'' operation is somehow regarded as the activation function of the neurons in the first hidden layer.
\end{remark}

\begin{remark}
There is a geometric interpretation: ``max - min'' operation measures the distance between the points $x$ and $w^{(1)}_j$.
Therefore $z(x)$ is invariant along one vectors $\mathbf{1}$.
\end{remark}

\begin{remark}
There are alternative ways to attain the invariance such as
\begin{equation}
z_j = \max_i(x_i + w^{(1)}_{ji}) - \textrm{2nd} \max_i(x_i + w^{(1)}_{ji}).
\end{equation}
There is a geometric interpretation: ``max - 2nd max'' operation measures the distance between a point $x$ and the tropical hyperplane whose normal vector is $w^{(1)}$ \cite{joswigBook}.
Therefore $z(x)$ is invariant along one vectors $\mathbf{1}$.
You could even use $j$-th max in general.
However, the repertoire of functions never increase by using alternative ones.
That is, from the view point of universal approximation theorem, using Eq.~\eqref{eq_activation_function} suffices.
In addition, Eq.~\eqref{eq_activation_function} seems to perform better than the alternative ones according to our numerical experiments (not shown).
Therefore we solely use Eq.~\eqref{eq_activation_function} as a tropical embedding layer in what follows.
\end{remark}

\begin{remark}
    Suppose $A \in \mathbb{Z}_+^{N \times d}$.  Then we consider the ReLU such that
    \[
    \max\{Ax + b, 0\}.
    \]
    Assume that $A {\bf 1} \not = 0$.
    Suppose $x \in \mathbb{R}^d/\mathbb{R}{\bf 1}$.  Then we have $x := x + c\cdot (1, \ldots , 1) = x + c \cdot {\bf 1}\in \mathbb{R}^d/\mathbb{R}{\bf 1}$.
    Then for $c \ll 0$ and fixed $x$, we have:
    \[
    \begin{array}{ccl}
        \max\{Ax + cA{\bf 1} + b, 0\} & = & 0.\\
    \end{array}
    \]
    As $c \to -\infty$, we have
    \[
    \frac{1}{1+\exp (-\max\{Ax + cA{\bf 1} + b, 0\})} \to \frac{1}{1 + 1} = 1/2
    \]
    for any $x \in \mathbb{R}^d/\mathbb{R}{\bf 1}$.  
    Also for $c \gg 0$ and fixed $x$, we have:
    \[
    \begin{array}{ccl}
        \max\{Ax + cA{\bf 1} + b, 0\} & = & Ax + cA{\bf 1} + b.\\
    \end{array}
    \]
    As $c \to \infty$, we have
    \[
    \frac{1}{1+\exp (- (Ax + cA{\bf 1} + b))} \to 1
    \]
    for any $x \in \mathbb{R}^d/\mathbb{R}{\bf 1}$. 
    Therefore, neural networks with the ReLU cannot learn from observation in these cases.  
    However, with the activator function defined in Eq.~\eqref{eq_activation_function}, we have 
    \[
    \max_i(x_i + c\cdot {\bf 1}+ w^{(1)}_{ji}) - \min_i(x_i + c\cdot {\bf 1} + w^{(1)}_{ji}) = \max_i(x_i + w^{(1)}_{ji}) - 
    \min_i(x_i + w^{(1)}_{ji}).
    \]
\end{remark}

\begin{remark}
    Classical neural networks are 
    not well-defined in the tropical
    projective torus, since the
    neuron values are not invariant 
    under transformations 
    of the form $x \to x + (c,\dots,c)$. 
    Meanwhile, 
    the tropical embedding layer
    of Eq.~\eqref{eq_activation_function}
    is
    invariant under such transformations.
\end{remark}

\section{Universal Approximation Theorems for Tropical Neural Networks}
It is very important to check if the tropical embedding layer as in Eq.~\eqref{eq_activation_function} followed by classical layers has enough varieties to represent considerable input-output relations \cite{Calin20}.
In this section, we show that the tropical neural network can approximate enough variety of functions so that we can safely use it.

\begin{definition}
The norm $\| \cdot \|_q$ for $q \ge 1$ is defined by
\begin{equation}
\| f(x) \|_q = \int_{\mathbb{R}^n}|f(x)|^q dx
\end{equation}
The space $L^q(\mathbb{R}^d), (1 < q < \infty),$ is the set of Lebesgue integrable functions $f$ from $\mathbb{R}^d$ to $\mathbb{R}$ for which $\| f(x) \|_q < \infty$.
\end{definition}

\begin{definition}
The space $C^0(\mathbb{R}^d)$ is the set of continuous, compactly suppported functions from $\mathbb{R}^d$ to $\mathbb{R}$.
\end{definition}

\begin{remark}
Note that $C^0(\mathbb{R}^d) \subset L^q(\mathbb{R}^d)$.
\end{remark}

For the classical case, a universal approximation theorem for ReLU feedforward neural nerworks has been proved in \cite{Arora18}.

\begin{theorem}[classical universal approximation theorem \cite{Arora18}]
\label{universal_approximation_theorem1}
Any function of $x_j$ for $j=1,\ldots,d$ in $L^q (\mathbb{R}^d), (1 < q < \infty),$ can be arbitrarily well approximated in the $\| \cdot \|_q$ by a ReLU feedforward neural network with at most $L=2(\lfloor \log_2 d \rfloor + 2)$ layers. 
\end{theorem}

As the $d-1$ neurons in the tropical embedding layer can easily represent $(x_j - x_d)$ for $j=1,\ldots,d-1$ and Theorem \ref{universal_approximation_theorem1} can be applied to the second and later layers of a tropical neural network (that is equivalent to a classical neural network), we can prove the following theorem.

\begin{theorem}[tropical universal approximation theorem]
Any function of $(x_j - x_d)$ for $j=1,\ldots,d-1$ in $L^q(\mathbb{R}^d/\mathbb{R}\mathbf{1}) \simeq L^q(\mathbb{R}^{d-1}), (1 < q < \infty),$ can be arbitrarily well approximated in the $\| \cdot \|_q$ by a tropical neural network with at most $L=2(\lfloor \log_2 d \rfloor + 2)+1$ layers (which include an tropical embedding layer as the first layer).
\end{theorem}
\begin{proof}
For any $f \in L^q(\mathbb{R}^{d-1})$, $\exists g \in C_0(\mathbb{R}^{d-1})$ such that $\| f-g \|_q < \epsilon/2$ \cite{Calin20}.
Let $K$ be the support of $g$ and let $M$ be $\max_{x \in K} \|x\|$.
For $x \in K$, we can set $w^{(1)}_{jj} = -w^{(1)}_{jd} = 2M$ and $w^{(1)}_{ji} = 0$ for $i \neq j, d$ to obtain $z_j = x_j - x_d + 4M$ for $j=1,\ldots,d-1$.
This means that a neuron in the first tropical embedding layer can represent $x_j - x_d$.
Then $d-1$ neurons can represent $z_1$, $z_2$, \ldots, $z_{d-1}$.
Finally, simply apply Theorem \ref{universal_approximation_theorem1} to the classical neural network $F(z_1,\ldots,z_{d-1})$ consisting of the second and later layers of a tropical neural network to obtain $\| g-F \|_q < \epsilon/2$.
Taken together, $\| f-F \|_q < \| f-g \|_q + \| g-F \|_q < \epsilon$.
\end{proof}

There is another type of classical universal approximation theorems.

\begin{definition}
The width $d_m$ of a neural network is defined to be the maximal number of neurons in a layer.
\end{definition}

\begin{theorem}[classical universal approximation theorem for width-bounded ReLU networks \cite{Lu17}]
\label{universal_approximation_theorem2}
For any $f \in L^1 (\mathbb{R}^d)$ and any $\epsilon>0$, there exists a classical neural networks $F(x)$ with ReLU activation functions with width $d_m \le d + 4$ that satisfies
\begin{equation}
    \int_{\mathbb{R}^d} |f(x)-F(x)| dx < \epsilon .
\end{equation}
\end{theorem}

Again, as the $d-1$ neurons in the tropical embedding layer can easily represent $(x_j - x_d)$ for $j=1,\ldots,d-1$ and Theorem \ref{universal_approximation_theorem2} can be applied to the second and later layers of a tropical neural network (that is equivalent to a classical neural network), we can prove the following theorem.

\begin{theorem}[tropical universal approximation theorem with bounded width]
For any function $f$ of $(x_j - x_d)$ for $j=1,\ldots,n-1$ in $L^1(\mathbb{R}^d/\mathbb{R}\mathbf{1}) \simeq L^1(\mathbb{R}^{d-1})$ and any $\epsilon>0$, there exists a tropical neural networks $F(x)$ with width $d_m \le d + 4$ that satisfies
\begin{equation}
    \int_{\mathbb{R}^{d-1}} |f(x)-F(x)| dx < \epsilon .
\end{equation}
\end{theorem}
\begin{proof}
For any $f \in L^1(\mathbb{R}^{d-1})$, $\exists g \in C_0(\mathbb{R}^{d-1})$ such that $\| f-g \|_1 < \epsilon/2$ \cite{Calin20}.
Let $K$ be the support of $g$ and let $M$ be $\max_{x \in K} \|x\|$.
For $x \in K$, we can set $w^{(1)}_{jj} = -w^{(1)}_{jd} = 2M$ and $w^{(1)}_{ji} = 0$ for $i \neq j, d$ to obtain $z_j = x_j - x_d + 4M$ for $j=1,\ldots,d-1$.
This means that a neuron in the first tropical embedding layer can represent $x_j - x_d$.
Then $d-1$ neurons can represent $z_1$, $z_2$, \ldots, $z_{d-1}$.
Finally, simply apply Theorem \ref{universal_approximation_theorem2} to the classical neural network $F(z_1,\ldots,z_{d-1})$ consisting of the second and later layers of a tropical neural network to obtain $\| g-F \|_q < \epsilon/2$.
Taken together, $\| f-F \|_q < \| f-g \|_q + \| g-F \|_q < \epsilon$.
\end{proof}

\section{Backpropagation Rule for Simplest Tropical Neural Networks}

\begin{figure}[h]
 \centering
 \includegraphics[width=\textwidth]{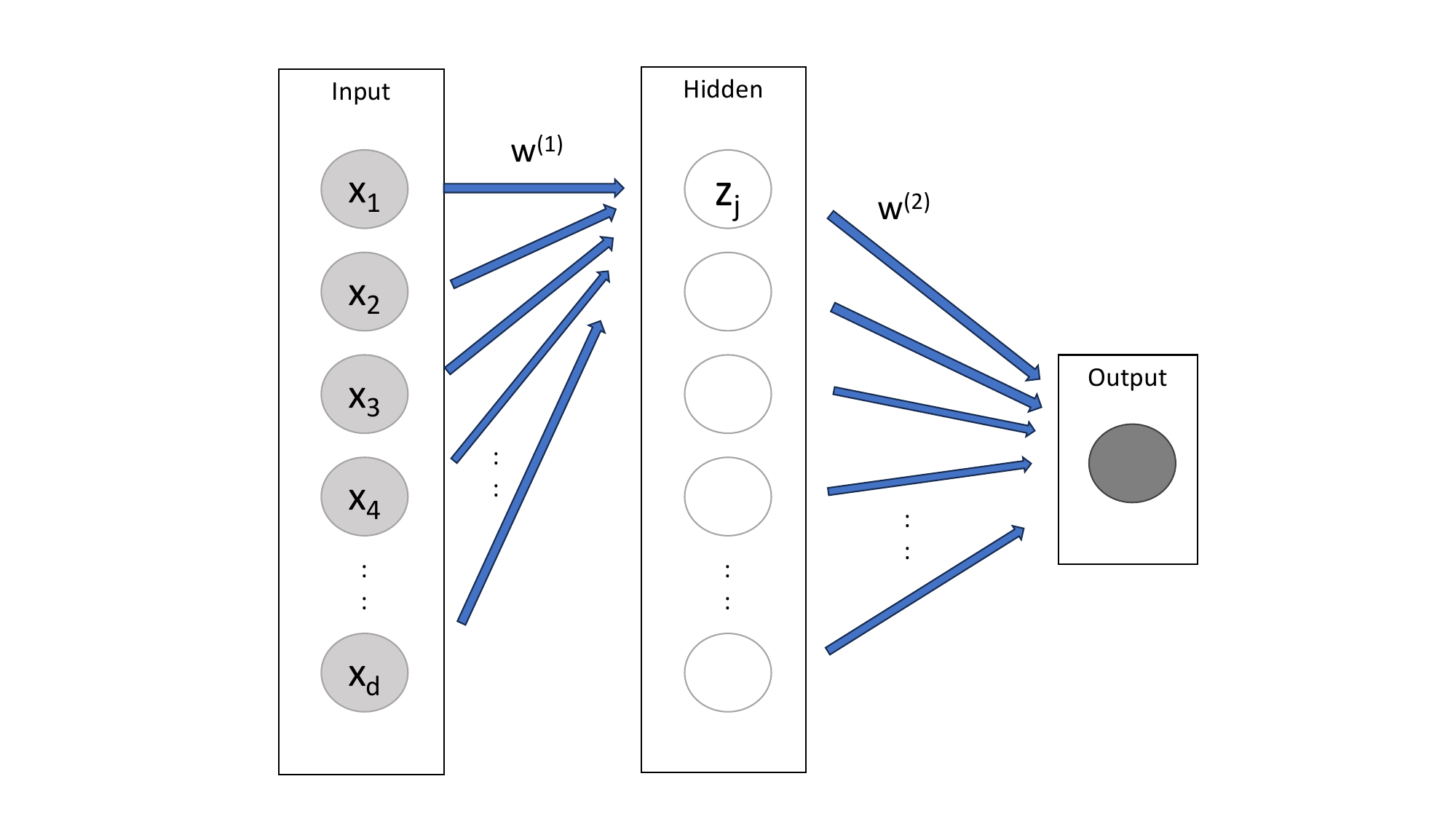}
 \caption{architecture of a simplest neural networks that accept a vector in $\mathbb{R}^d/\mathbb{R}\mathbf{1}$}
 \label{fig_design_simplest}
\end{figure}

Here we demonstrate that the gradients of the loss function with respect to weights exist for tropical neural networks.
The gradient is computable through the chain rule for differentials called backpropagation rule in the similar fashion to the classical case.
The gradients obtained in this way can guarantee the successful update of the weights at each iteration of learning.

We consider a simplest three layer network whose weights in the first and the second layers are denoted as $w^{(1)} \in \mathbb{R}^{d \times N}$ and $w^{(2)} \in \mathbb{R}^{N \times 1}$.
Suppose the activity in the first hidden layer is given by Eq.~\ref{eq_activation_function} and the output of the network is given as
\begin{equation} \label{eq:response}
y = \sum_j^N w^{(2)}_j z_j .
\end{equation}
Note that although here we derive a backpropagation rule for this regression setting just for simplicity, the backpropagation rule can be derived in a similar manner for the classification setting with a sigmoid function, too.

Below is the summary of parameters of the neural network.
\begin{itemize}
  \item $w^{(1)}$, $w^{(2)}$: weights in the first and second layers;
  \item $z_j$: activation of $j$-th neuron in the hidden layer; and
  \item $y$: activation of the output neuron.
\end{itemize}

\begin{theorem}
The partial derivatives of the cost function $Q:=\frac{1}{2}(y-y^{\textrm{true}})^2$ with respect to weights for the above tropical neural network $y=f(x)$ are given by
\begin{equation}
\frac{\partial Q}{\partial w^{(1)}_{ji}} = (y-y^{\textrm{true}}) w^{(2)}_j (\delta(i=i^\textrm{max}_j)-\delta(i=i^\textrm{min}_j)),
\label{eq_gradients}
\end{equation}
where $i^\textrm{max}_j$ (or $i^\textrm{min}_j$) is the index $i$ for which $(x_i + w^{(1)}_{ji})$ takes the maximum (or minimum) and
\begin{equation}
\frac{\partial Q}{\partial w^{(2)}_j} = (y-y^{\textrm{true}})z_j.
\end{equation}
\end{theorem}
\begin{proof}
Direct calculations.
\end{proof}

\begin{example}
\label{ex_delta}
As a simplest example of Eq.~\eqref{eq_gradients}, let us consider the three dimensional input case $(d=3)$.
Suppose the number of neurons in the middle layer is one and its activity is z, for simplicity.
Assume $x_1 = 1$, $x_2 = 2$ and $x_3 = 3$ and $w^{(1)}_1 = w^{(1)}_2 = w^{(1)}_3 = 0$.
Then $(x_i + w^{(1)}_1) < (x_i + w^{(1)}_2) < (x_i + w^{(1)}_3)$
and $i_\textrm{max}=3$ and $i_\textrm{min}=1$.
Therefore,
\begin{equation}
\frac{\partial Q}{\partial w^{(1)}_i} = \left\{
\begin{array}{ll}
  -(y-y^{\textrm{true}}) w^{(2)} & (i=1) \\
  0 & (i=2)\\
  (y-y^{\textrm{true}}) w^{(2)} & (i=3).
\end{array}
\right.
\end{equation}
In this case we have $z=3-1=2$
If, furtheremore, $w^{(2)}=1$, then, $y=w^{(2)}z=2$ and
\begin{equation}
\frac{\partial Q}{\partial w^{(1)}_i} = \left\{
\begin{array}{ll}
  -(2-y^{\textrm{true}}) , & (i=1) \\
  0 , & (i=2)\\
  2-y^{\textrm{true}} . & (i=3),
\end{array}
\right.
\end{equation}
$w^{(1)}_i$ can be, for example, updated by the SGD rule: $\Delta w^{(1)}_{i} = - \eta \frac{\partial Q}{\partial w^{(1)}_i}$, where $\eta > 0$ is a learning rate.
Then, $w^{(1)}_3$ increases (and $w^{(1)}_1$ decreases) if $2>y^{\textrm{true}}$.
\end{example}

\begin{remark}
It is interesting that only two of $w^{(1)}_{i}$ are modified while the others remains.
Note that $\Delta w^{(1)}$ is orthogonal to the one vector, $\mathbf{1}:= (1, 1, \ldots, 1) \in \mathbb{R}^d$.
It is interesting to elucidate how this learning rule works as a dynamical system.
\end{remark}

\section{TensorFlow2 Codes for Tropical Neural Networks}
In order to boost computing with GPUs, we implement tropical neural networks in TensorFlow 2 \cite{Chollet21}.
As is the case for the classical neural networks, the auto-differential is the key for the GPU implementation of tropical neural networks.
In order to guarantee fast auto-differentials, all the calculations must be implemented only with the math functions in TensorFlow2 such as top\_k($v$,$d$), which returns the maximum and the minimum of a vector $v$.

In practice, it is essential to create a user-friendly class for the tropical embedding as the first layer of the tropical neural networks, that is scalable for big data.
The following code defines a hand-made class called TropEmbed(), which enables us to easily implement the tropical neural networks in the Keras/Tensorflow style.

\begin{lstlisting}[basicstyle=\ttfamily\footnotesize, frame=single]
class TropEmbed(Layer):
    def __init__(self, units=2, input_dim=3):
        super(TropEmbed, self).__init__()
        self.w = self.add_weight(shape=(units, input_dim), \
                initializer="random_normal")
        self.units = units
        self.input_dim = input_dim

    def call(self, x):
        x_reshaped = tf.reshape(x,[-1, 1, self.input_dim])
        Bcast = repeat_elements(x_reshaped, self.units, 1)
        val, i = tf.math.top_k(Bcast + self.w, self.input_dim)
        return val[:,:,0] - val[:,:,-1]

# usage
model = Sequential([TropEmbed(10, d), Dense(1)])
\end{lstlisting}
%
%
%

The codes for TropEmbed() class and for reproducing all the figures in this paper are available at \url{https://github.com/keiji-miura/TropicalNN}.

\section{Weight Initialization Based on Extreme Value Statistics}
Weight initializations are important for avoiding the divergence and vanishment of neural activities after propagating many layers.
For the classical neural networks, Xavier's and He's initializations are famous \cite{Glorot10, He16}.
Here we consider a tropical analogue.

\begin{definition}[Generalized Hilbert Projective Metric]
\label{eq:tropmetric}
For any points $v:=(v_1, \ldots , v_d), \, w := (w_1, \ldots , w_d) \in \mathbb R^d \!/\mathbb R {\bf 1}$,  the {\em tropical distance} (also known as {\em tropical metric}) $d_{\rm tr}$ between $v$ and $w$ is defined as:
\begin{equation*}
d_{\rm tr}(v,w)  := \max_{i \in \{1, \ldots , d\}} \bigl\{ v_i - w_i \bigr\} - \min_{i \in \{1, \ldots , d\}} \bigl\{ v_i - w_i \bigr\}.
\end{equation*}
\end{definition}

\begin{lemma}
\label{lm:weight}
Suppose $x_i, w_i \sim N(0,1)$ for $i=1,\ldots,d$. Then the expectation and variance of $d_{\rm tr}(x,-w)$ can be approximated by $2 \sqrt{2} (a_d \gamma + b_d)$ and $\frac{\pi^2}{3 \log d}$, respectively, where $a_d =  \frac{1}{\sqrt{2 \log d}}$ and $b_d = \sqrt{2 \log d} - \frac{\log \log d + \log (4\pi)}{2\sqrt{2 \log d}}$.
\end{lemma}
\begin{proof}
As $x_i + w_i \sim N(0,2)$, $Z := \frac{\max\{x+w\}/\sqrt{2}-b_d}{a_d}
\sim Gumbel(0,1)$ as $d \to \infty$.
Therefore, $\mathrm{Ex}[d_{\rm tr}(x,-w)] = \mathrm{Ex}[2\max\{x+w\}] \xrightarrow[d\to\infty]{} 2 \sqrt{2}( a_d \mathrm{Ex}[Z] + b_d )$.
$\mathrm{Var}[d_{\rm tr}(x,-w)] = \mathrm{Var}[\max\{x+w\}-\min\{x+w\}] = 2 \mathrm{Var}[\max\{x+w\}] + 2 \mathrm{Cov}[\max\{x+w\},-\min\{x+w\}] 
\xrightarrow[d\to\infty]{}
2 \times 2 a_d^2 \mathrm{Var}[Z]
= 2 \times 2 a_d^2 \frac{\pi^2}{6}$ where $\mathrm{Cov}[\max\{x+w\},-\min\{x+w\}] \xrightarrow[d\to\infty]{} 0$ was assumed.
\end{proof}

Here we confirm that the above scaling holds actually by numerical calculations.

\begin{figure}[h]
 \centering
 \includegraphics[scale=0.4]{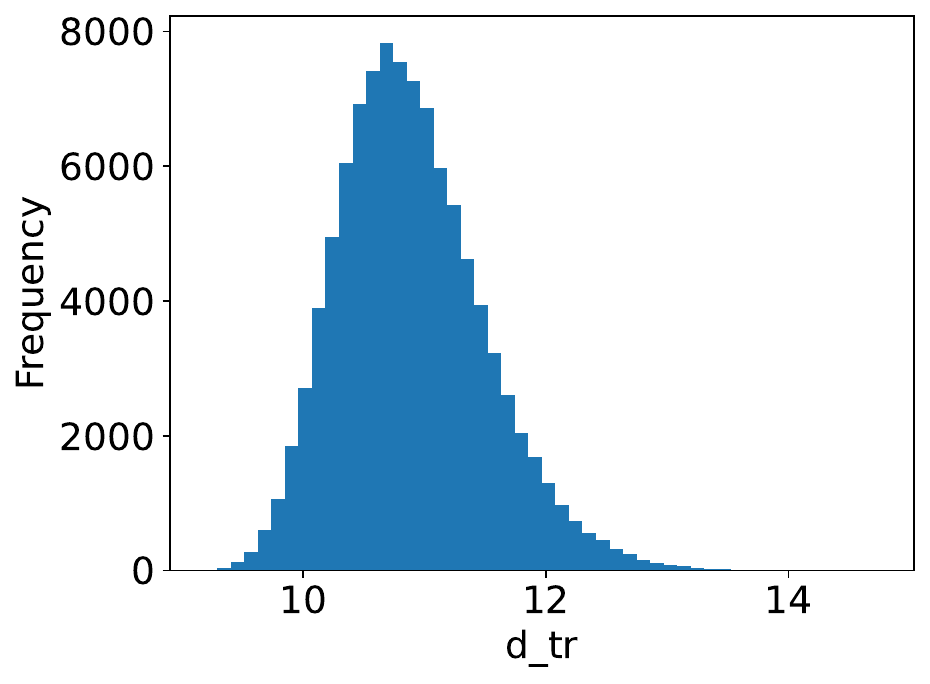}
 \caption{histogram of simulated $d_{\rm tr}(x,-w)$ for the same situation as in Lemma \ref{lm:weight} with $d=10000$. For the histogram, $100000$ samples of $d_{\rm tr}(x,-w)$ are used. The simulated mean is $10.893$ while the theoretical prediction is $10.954$. The simulated std is $0.604$ while the theoretica prediction is $0.598$. The mean and std are the same as predicted from the theory.}
 \label{fig_hist_dtr}
\end{figure}

One way for better weight initialization is to choose the scale of $w$ so that the variance of the neural activity in the embedding layer becomes $1$.

\begin{lemma}
Suppose $x_i \sim N(0,1)$ and $w_i \sim N(0,\frac{6 \log d}{\pi^2} -1)$ for $i=1,\ldots,d$.
Then the expectation and variance of $d_{\rm tr}(x,-w)$ can be approximated by
$2 \sqrt{\frac{6 \log d}{\pi^2}} (a_d \gamma + b_d)$ and $1$, respectively.
\end{lemma}
\begin{proof}
As $x_i + w_i \sim N(0,\frac{6 \log d}{\pi^2})$, $Z := \frac{\max\{x+w\}/\sqrt{\frac{6 \log d}{\pi^2}}-b_d}{a_d} \sim Gumbel(0,1)$.
Therefore, $\mathrm{Ex}[d_{\rm tr}(x,-w)] = \mathrm{Ex}[2\max\{x+w\}] \xrightarrow[d\to\infty]{} 2 \sqrt{\frac{6 \log d}{\pi^2}}( a_d \mathrm{Ex}[Z] + b_d )$.
$\mathrm{Var}[d_{\rm tr}(x,-w)] \xrightarrow[d\to\infty]{} 2 \times \frac{6 \log d}{\pi^2} a_d^2 \mathrm{Var}[Z] = 2 \frac{6 \log d}{\pi^2} a_d^2 \frac{\pi^2}{6} = 1$.
\end{proof}

To control the standard deviation of the weights, you can customize an initializer (instead of simply specifying initializer $=$ "random\_normal") in the definition of TropEmbed class.
\begin{lstlisting}[basicstyle=\ttfamily\footnotesize]
> ini = tf.keras.initializers.RandomNormal(mean=0., stddev=1.)
> self.w = self.add_weight(shape=(units, input_dim), \
      initializer=ini)
\end{lstlisting}
However, as the weght initialization should be done together with the data preprocessing, in this paper we entirely use the default value of stddev=0.05 for "random\_normal" for simplicity.
\begin{lstlisting}[basicstyle=\ttfamily\footnotesize]
> self.w = self.add_weight(shape=(units, input_dim), \
      initializer="random_normal")
\end{lstlisting}

\section{Computational Experiments}

In this section, we apply tropical neural networks with one hidden layer (the tropical embedded layer) to simulated data as well as empirical data.  Then later we compare its performance against neural networks with one hidden layer with ReLU activator.  

\subsection{Small simulated data}
First we illustrate our tropical neural networks with one hidden layer with 16 neurons and with one output Sigmoid function using a small example.
First we generate two dimensional $16+16$ random points from the Gaussian distributions with means $(0.5, -0.5)$ and $(-0.5, 0.5)$ with the unit covariance matrix.
Then these points are randomly translated by $(c, c)$ where $c$ is a Gaussian random variable whose standard deviation is 4.
The left and right figures in Figure \ref{fig:predPro} show the actual test labels and the predicted probabilities of the test data by the tropical neural networks with one hidden layer with 16 neurons.

\begin{figure}[h!]
    \centering
    \includegraphics[width=1\textwidth]{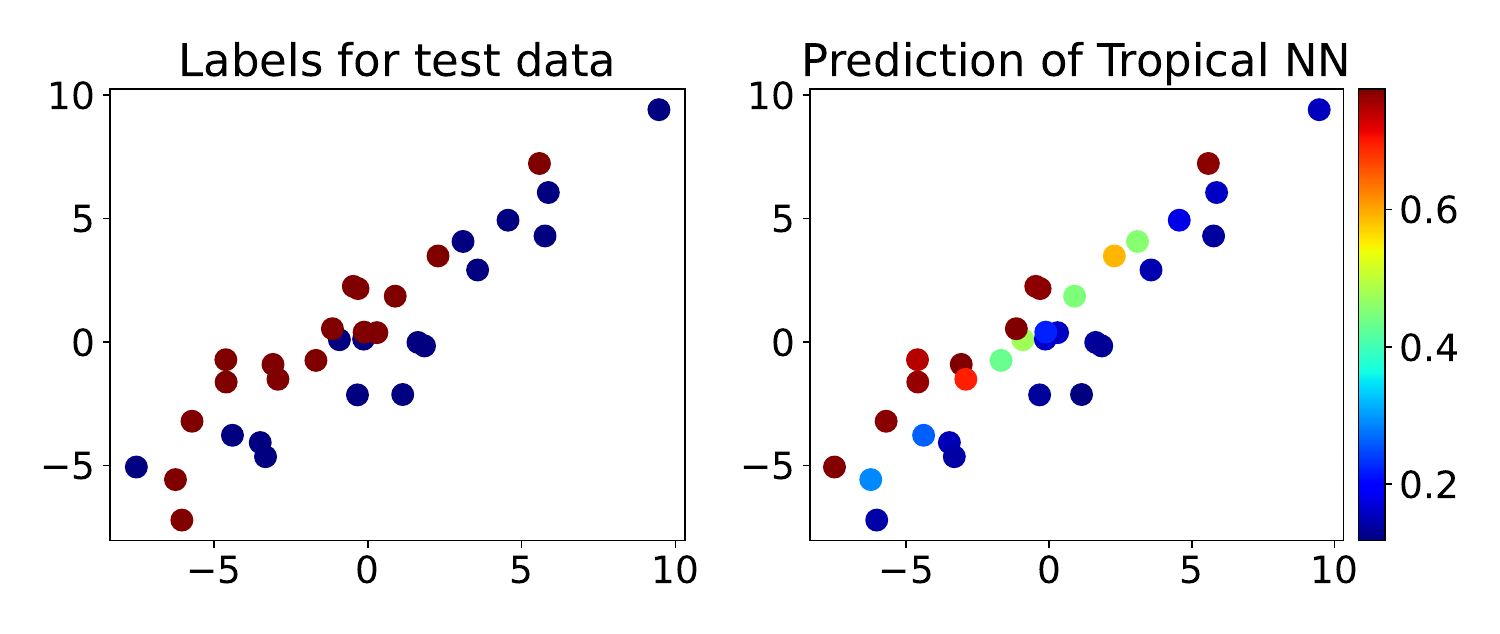}
    \caption{Predicted probabilities by the tropical neural networks on a small example.}
    \label{fig:predPro}
\end{figure}


\subsection{High-dimensional simulated data}
Second we demonstrate that a tropical neural network with one hidden layer with 8 neurons and with one output sigmoid function works against the curse of dimensionality, where the most of the variables in this high dimensional data are rather just noises \cite{YTMM}.
We generate $d$ dimensional $16+16$ random points from the Gaussian distributions with means $(0.5, -0.5, 0, \ldots, 0)$ and $(-0.5, 0.5, 0, \ldots, 0)$ with the unit covariance matrix.
Then these points are randomly translated by $(c, c, c, \ldots, c)$ where $c$ is a Gaussian random variable whose standard deviation is 6.
The result demonstrates that the tropical neural networks work robustly against the curse of dimensionality.

\begin{figure}[h!]
    \centering
    \includegraphics[width=0.6\textwidth]{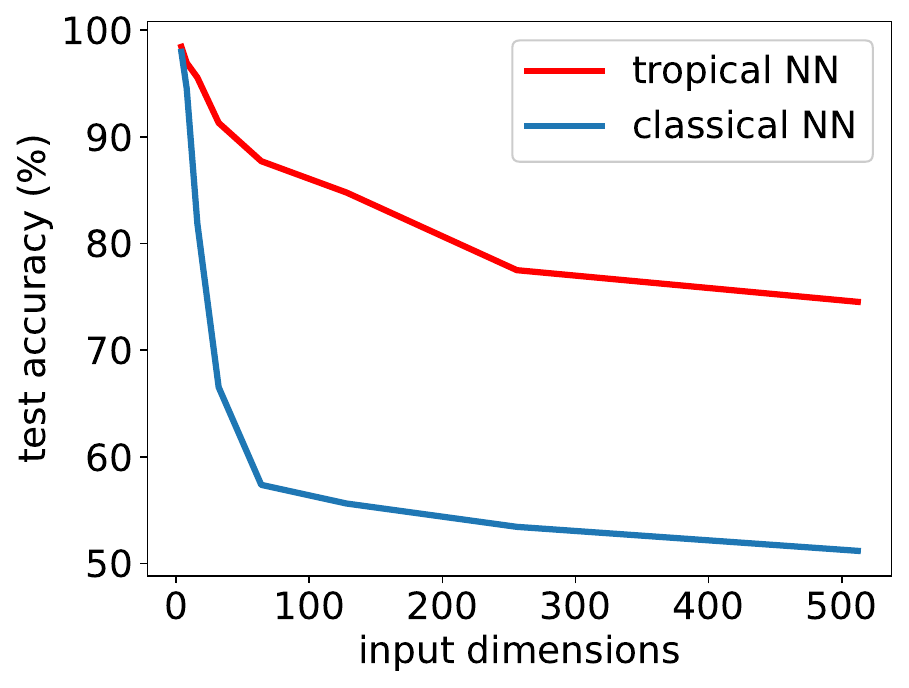}
    \caption{Application of the tropical neural networks to a high-dimensional example. The test accuracy averaged over 100 trials is plotted. The tropical neural networks work robustly against the curse of dimensionality.}
    \label{fig:predPro2}
\end{figure}

\subsection{simulated data generated from the multi-species coalescent model}

In this subsection we apply the tropical neural networks to a sample of phylogenetic trees generated under the {\em multi-species coalescent model}. 

A phylogenetic tree is a weighted tree whose leaves are labeled with $[m]:=\{1, 2, \ldots , m\}$, where $m$ is the number of leaves, and whose internal nodes are unlabeled. A weight on each edge of a phylogenetic tree is considered as a distance from a node to another node on the tree and in evolutionary biology, a weight on an edge can be considered as a product of an evolutionary time and mutationa rate \cite{SS2}. On this paper we consider a rooted phylogenetic tree with a leaf label set $[m]$.  An equidistant tree with $m$ leaves is called {\em equidistant tree} if the total weights on the unique path from its room to each leaf is the same for all leaf in $[m]$. Under the multi-species coalescent model which can be used to analyze gene trees, which are phylogenetic trees reconstructed from genes in a genome, we assume that gene trees are all equidistant trees.  Therefore in this paper we assume that all phylogenetic trees are equidistant trees.  

To conduct a statistical analysis on a set of phylogenetic trees, we consider a {\em space of phylogenetic trees} with fixed $[m]$. A space of phylogenetic trees on $[m]$ is a set of all possible phylogenetic trees with $[m]$ and it is well-known that it is not Euclidean \cite{SS}. It is also well-known that the space of all possible equidistant trees on $[m]$ with the {\em tropical metric} under the max-plus algebra is a subspace of the tropical projective space \cite{AK,YZZ}.  In order to define the space of equidistant trees, first we define {\em ultrametrics}.  COnsider a map $u: [m] \times [m] \to \mathbb{R}$ such that $u(i, j) = u(j, i)$ and $u(i, i) = 0$.  This map is called a {\em dissimilarity map} on $[m]$.  If a dissimilarity map $u$ satisfies that
\[
\max\{u(i, j), u(i, k), u(j, k)\}
\]
achieve at least twice, then we call $u$ an {\em ultrametric}.  
\begin{example}\label{ex:ultra}
    Suppose $m = 3$, i.e., $[3] = \{1, 2, 3\}$ and suppose 
    \[
    u(1, 2) = u(2, 1) =  1, u(1, 3) = u(3, 1) = 1, u(2, 3) = u(3, 2) = 0.5
    \]
    and $u(i, i) = 0$ for all $i = 1, 2, 3$.  Since 
    \[
    \max\{u(1, 2), u(1, 3), u(2, 3)\} = 1
    \]
    and it achieves twice, i.e., $u(1, 2) = u(1, 3) = 1$.
    Thus, $u$ is an ultrametric.  
\end{example}
Consider dissimilarity maps $u_T$ on a phylogenetic tree $T$ on $[m]$ such that $u(i, j)$ is the total weights on the unique path from a leaf $i$ to a leaf $j$ for all $i, j \in [m]$.  Then we have the following theorem:
\begin{theorem}[\cite{Buneman}]
    Consider an equidistant tree $T$ on $[m]$.  Then $u_T$ realizes an equidistant tree $T$ on $[m]$ if and only if a dissimilarity map $u_T$ is ultrametric.
\end{theorem}
\begin{example}
    Suppose we have an ultrametric from Example \ref{ex:ultra}.  An equidistant tree whose dissimilarity maps are ultrametric in Example \ref{ex:ultra} is a rooted phylogenetic tree with leaves $[3]=\{1, 2, 3\}$ shown in Figure \ref{fig:ultra}.
    \begin{figure}[h!]
        \centering
        \includegraphics[width=0.7\textwidth]{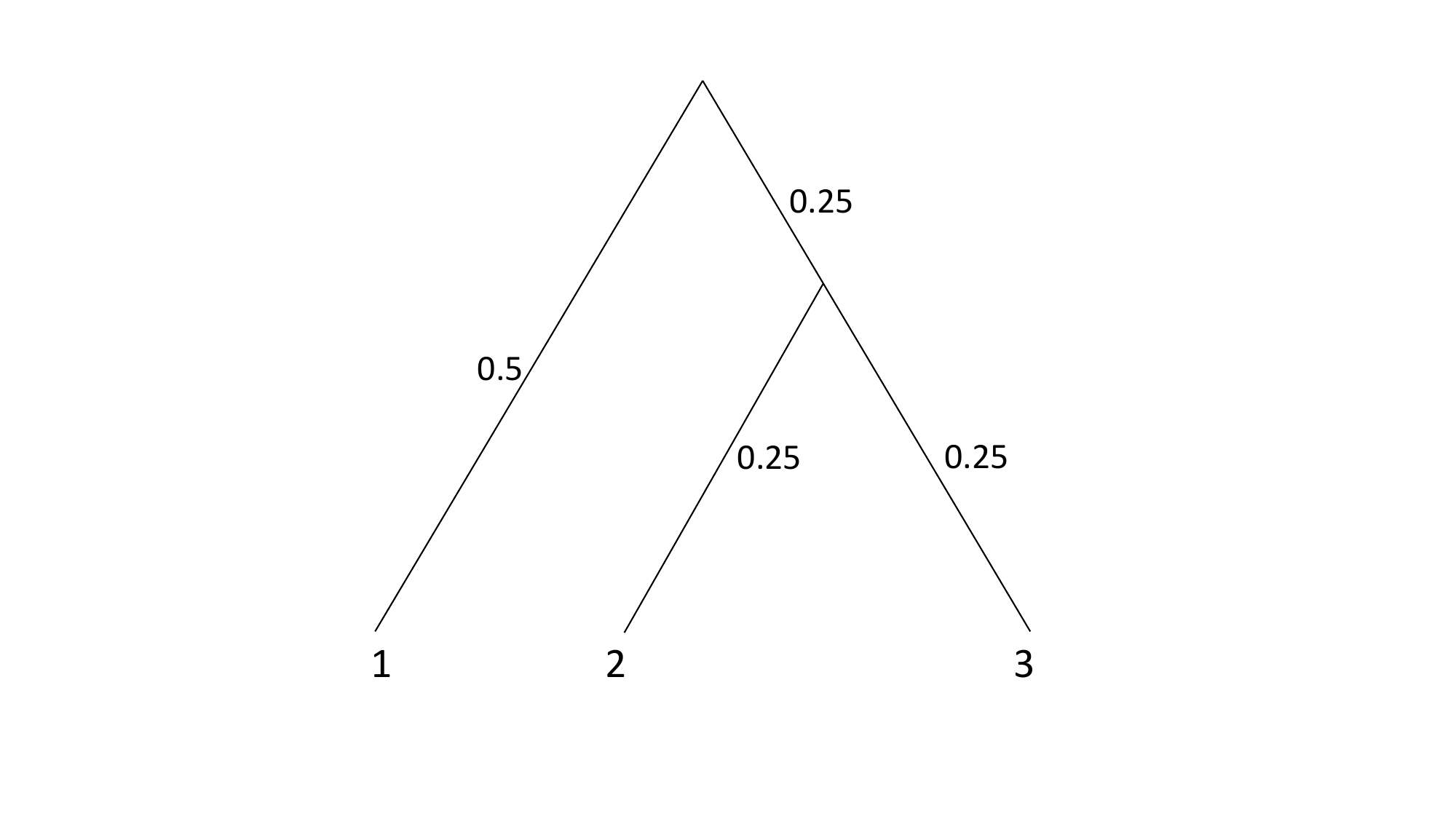}
        \caption{An equidistant tree with the dissimilarity maps which are ultrametric shown in Example \ref{ex:ultra}.}
        \label{fig:ultra}
    \end{figure}
\end{example}

Therefore, we consider the space of all possible ultrametrics on $[m]$ as the space of equidistant trees on $[m]$.  
Then we have the following theorem:
\begin{theorem}[\cite{AK}]
    The space of ultrametrics on $[m]$ is the tropicalization of the linear subspace defined by linear equations 
    \[
    x_{ij}-x_{ik}+x_{jk}=0
    \]
    for $1 \leq i \leq j \leq k \leq m$ by replacing sum with max operation and replacing multiplication by a classical summation.  
\end{theorem}
The space of ultrametrics on $[m]$ is a subspace of the tropical projective space $(\mathbb{R}^e\cup \{-\infty\}) / \mathbb{R} {\bf 1}$ where $e = \binom{m}{2}$. 
Therefore, we apply our method, tropical neural networks, to simulated data generated from the multi-species coalescent model using a software Mesquite \cite{mesquite}.  

The multi-species coalescent model has two parameters: species depth and effective population size. Here we fix the effective population size $N_e=100000$ and we vary
\[
R = \frac{SD}{N_e}
\]
where $SD$ is the species depth. We generate species trees using the Yule process.  Then we use the multi-species coalescent model to generate gene trees with a given species tree.  In this experiment, for each $R$, we generate two different set of $1000$ gene trees:  In each set of gene tree, we have a different species tree so that each set of gene trees is different from the other.  We conduct experiments with $R = 0.25, 0.5, 1, 2, 5, 10$.

Note that when we have a small ratio $R$, then gene trees become more like random trees since the species tree constrains less on tree topologies of gene trees.  Thus it is more difficult when we have a small $R$ while if we have a large $R$, then it is easier to classify since the species tree constrains more on tree topologies of gene trees.   

\begin{figure}[h!]
    \centering
    \includegraphics[width=\textwidth]{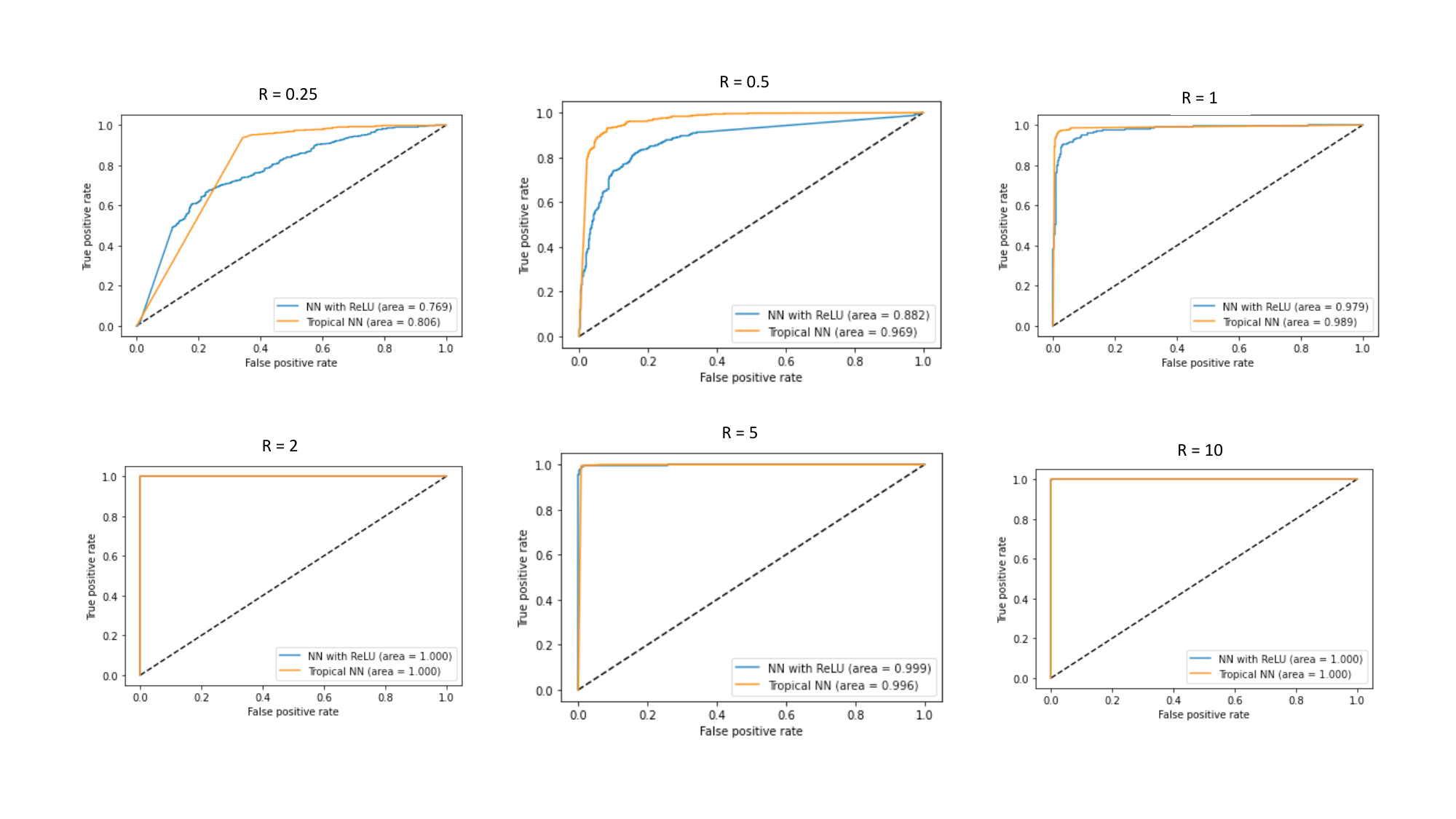}
    \vskip -0.3in 
    \caption{ROC curves for neural networks with ReLU and tropical neural networks with one hidden layer.  We conduct experiments with $R = 0.25, 0.5, 1, 2, 5, 10$.}
    \label{fig:roc}
\end{figure}

In this experiment, we set one hidden layer for each neural network: neural network with ReLU activators and neural network with tropical activators.  We set the Sigmoid function in the output node in both neural networks.  In each neural network, we set $1000$ neurons in the hidden layer.  

Figure \ref{fig:roc} shows ROC curves for neural networks with ReLU and tropical neural networks.  In general tropical neural networks perform better than neural networks with ReLU activator function.  

\subsection{Influenza data}
In this subsection we apply our method to genomic data for 1089 full length sequences of hemagglutinin (HA) for influenza A H3N2 from 1993 to 2017 collected in the state of New York obtained from the GI-SAID EpiFlu database (www.gisaid.org). These collected data were aligned using muscle developed by \cite{muscle} with the default settings. Then we apply the neighbor-joining method with the p-distance \cite{NJ} to reconstruct a tree from each sequenced data.  Each year corresponds to the first season. We also apply KDETrees \cite{KDETree} to remove outliers and a sample size of each year is about 20,000 trees. 

We apply tropical neural networks and neural networks with ReLU with one hidden layer with 10 neurons to all pairs of different years to see if they are significantly different one year to the other.
Heatmaps of accuracy rates with the 
probability threshold $0.5$ 
and AUC values
are shown in Figure \ref{fig:flu}.  
Again, tropical neural networks outperform classical neural networks.

\begin{figure}[h!]
    \centering
    \includegraphics[width=\textwidth]{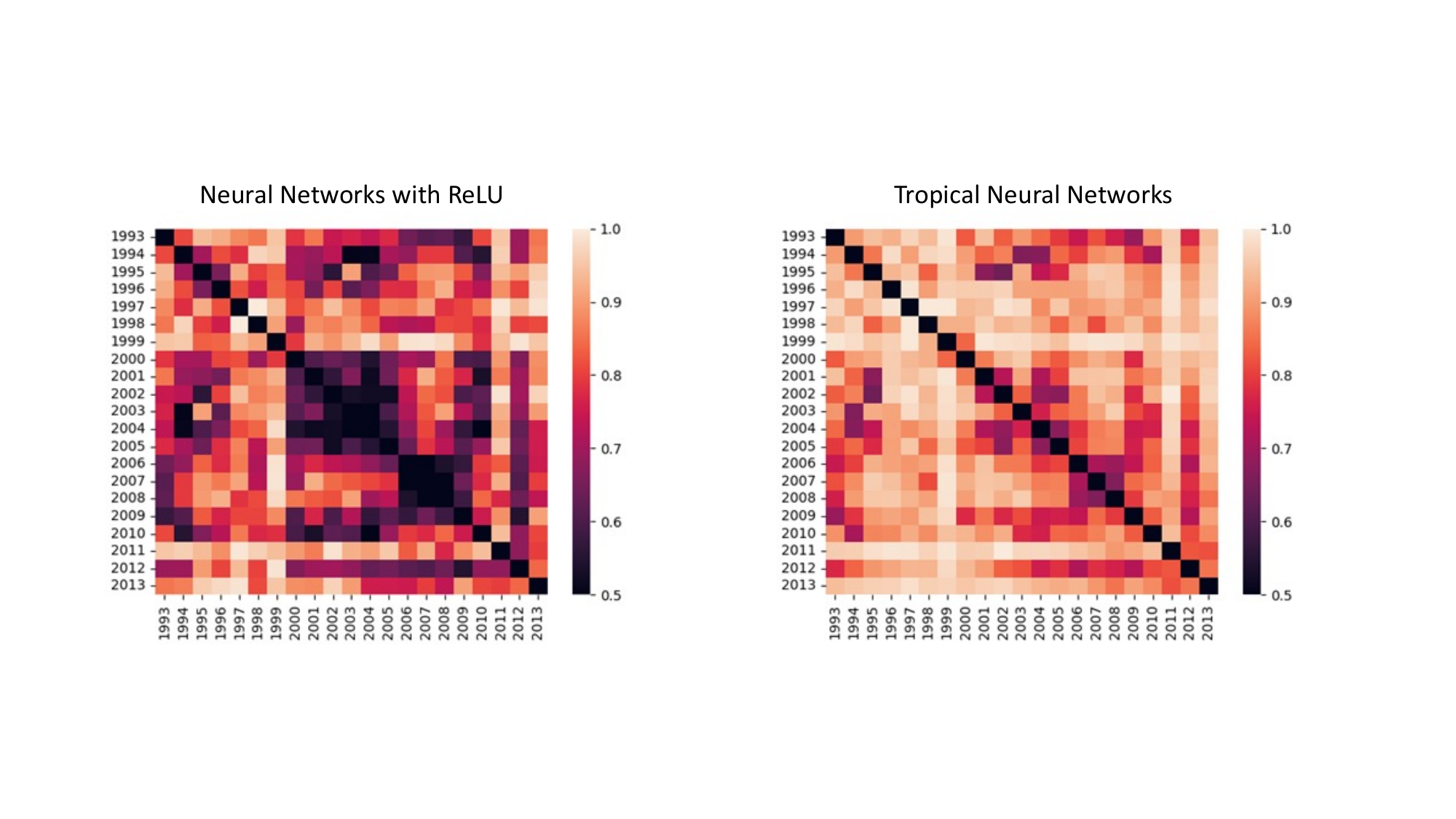}
    \begin{minipage}{.45\textwidth}
      \centering
    \includegraphics[width=\linewidth]{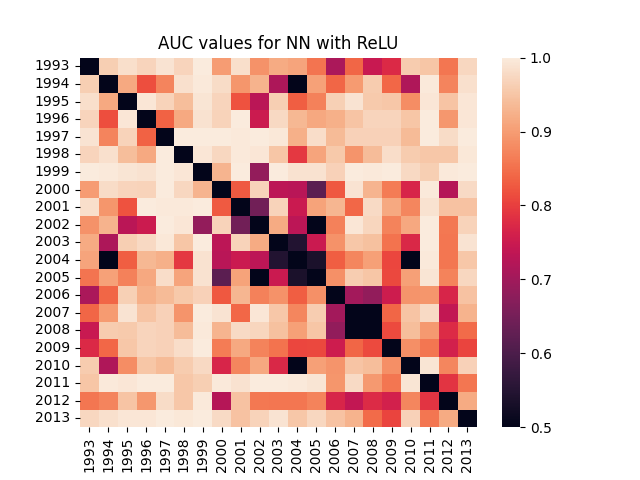}
      \label{fig:test1}
    \end{minipage}%
    \begin{minipage}{.45\textwidth}
      \centering
      \includegraphics[width=\linewidth]{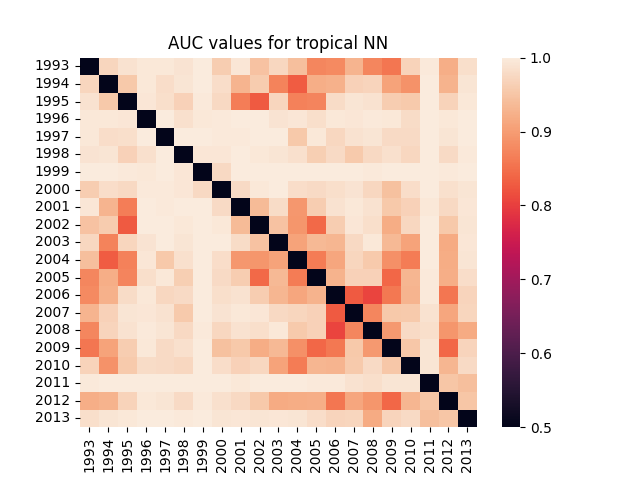}
      \label{fig:test2}
    \end{minipage}
    \caption{Heat maps for (top) classification rates 
    with threshold $0.5$
    and (bottom) AUC values 
    for classical neural networks with ReLU (left)
    and tropical neural networks (right). }
    \label{fig:flu}
\end{figure}



\subsection{Tropical Neural Network as a Generalization of Tropical Logistic Regression for Classification of Gene Trees}
A tropical neural network can be interpreted as a generalization of a tropical logistic regression.
The tropical logistic regression model
\cite{aliatimis2023tropical} is developed
for binary classification on gene trees and it has been 
proven to have higher predictive power 
than classical logistic regression to classify phylogenetic trees. 
It assumes that if $X$ is the covariate 
(gene tree), the binary response random 
variable is
$Y \sim \text{Bernoulli}(p(X))$, with 
\begin{equation*}
    p(x) = S\left(
        \lambda_0 d_{\rm tr}(x,\omega_0) 
        - \lambda_1 d_{\rm tr}(x,\omega_1) +
        C
    \right),
\end{equation*}
where $m$ is the number of leaves in each phylogenetic tree in the given sample, $e:= \binom{m}{2}$, $S$ is the sigmoid function, 
$\omega_0, \, \omega_1 \in 
\mathbb{R}^e/\mathbb{R}{\bf 1}$, and
$\lambda_0,\lambda_1,C \in \R$ with
$\lambda_0 \lambda_1 \geq 0$.
Note that this model is a special case for 
Eq.~\eqref{eq:response}, 
with the sigmoid function as the link, and
two neurons in the hidden
layer whose weights are 
$w_0^{(2)} = \lambda_0, 
w_1^{(2)}= -\lambda_1$ and
$w_0^{(1)} = -\omega_0, w_1^{(1)} = -\omega_1$.
Therefore, tropical logistic regression is 
almost identical to a tropical neural network
consisting of 
one tropical embedding layer with two neurons
and a classical layer, with the additional 
assumption that 
$w_0^{(2)} w_1^{(2)} \leq 0$. 

The one-species model described in \cite{aliatimis2023tropical} can be considered to be a neural network with $e$ neurons in the input layer, no hidden layers and a unique output neuron.
The activation function is the logistic function and the inner product used is tropical defined as
\begin{equation}
    \label{eq:tropical_inner_product}
    \langle x, -\omega \rangle :=
    d_{\rm tr}(x,\omega) - C,
\end{equation}
where $C$ can be considered to be a bias variable, 
similarly to the intercept variable in classical 
models.
Tropical logistic regression returns the sigmoid of the tropical inner product. We define the tropical generalised linear model as an extension to tropical logistic regression, where instead of the sigmoid function, we may use a different link/activation function.
If there are multiple outputs (multivariate generalized linear model (GLM))
and if we treat the output
layer as the new input layer and iterate this $L$ 
times, then we have an $L$-layer neural network.
In the same way that classical neural networks are 
a stack/recursive application of classical multivariate GLMs, 
tropical neural networks can be 
a stack of tropical multivariate GLMs. Effectively, all is 
identical to classical networks, but instead of 
applying classical inner products, we apply 
tropical inner products as defined in 
Eq.~\eqref{eq:tropical_inner_product}.
The $i$-th neuron of the $l$-th layer is defined as 
$x^{(l)}_i$ and computed through the recursive formula, 
\begin{equation}
    \label{eq:neuron_value_recursive}
    x^{(l)}_i = d_{\rm tr}
    \left(x^{(l-1)}_i, \omega^{(l)}_i \right) 
    - C^{(l)}_i,
\end{equation}
where 
$\Omega^{(l)} = (\omega^{(l)}_1, \omega^{(l)}_2, \dots,  \omega^{(l)}_{N_l}) \in \mathbb{R}^{ N_{l-1} \times N_l}$ is the weight matrix between layer $(l-1)$ and $l$
for the number $N_s$ of neurons in layer $s$,
and $C^{(l)} \in \mathbb{R}^{N_{l}}$. 
By assuming that all neurons share the same 
bias variable $c = C^{(l)}_i$ for all 
$i \in [N_l]$, 
Eq.~\eqref{eq:neuron_value_recursive}
reduces to  
Eq.~\eqref{eq_activation_function}, since 
vectors are defined up to an additive 
constant vector $(c,\dots,c)$ 
in the tropical projective torus.
When the last tropical embedding layer 
connects to the first classical layer, 
the constant bias vector is incorporated in the 
bias term of the classical layer.
Hence, tropical bias terms are redundant and
not considered in the development of 
tropical neural networks.
Thus, the tropical neural network which we propose in
this paper follows naturally as an extension 
of the tropical logistic regression model.

\section{Summary and Discussion}
In this paper, we first developed a tropical embedding layer.
We used the tropical embedding layer as the first layer of the classical neural networks to keep the invariance in the tropical projective space.
To check if this tropical neural network has enough flexibility, we next proved that this tropical neural network is a universal approximator.
After we derived a backpropagation rule for the tropical neural networks, we provided TensorFlow 2 codes for implementing a tropical neural network in the same fashion as the classical one, where the weights initialization problem is considered according to the extreme value statistics.
Finally we showed some applications as examples.

The tropical neural networks with the tropical metric worked better than the classical neural networks when the input data are phylogenetic trees which is included in the tropical projective torus. This is partly because only the tropical neural network can keep the invariance of the input vector under the one-vector which is innate in the tropical projective torus.

One of the nice properties of tropical neural networks is its tractability and interpretability in analysis.
The tropical embedding can be interpreted as taking the tropical distance to a point in the space of the tropical projective torus.
The activities of neurons in the tropical neural networks with the randomized weights and inputs can be analyzed by using the extreme value statistics.
The backpropagation rule of the tropical neural networks can be derived and interpreted rather easily.

The TensorFlow 2 codes for the Python class for tropical embedding was provided in the paper.
This makes it possible to implement a tropical neural network in the same familiar fashion as the classical one.
This facilitates, for example, to compare tropical and classical neural networks for the same data by using a common code.

Recent work shows that neural networks are vulnerable against adversarial attacks (for example, \cite{10.1007/978-3-642-40994-3_25,journals/corr/NguyenYC14,DBLP:journals/corr/SzegedyZSBEGF13,DBLP:conf/iclr/MadryMSTV18}).  However, our initial computational experiments on image data from computer vision show that tropical neural networks are robust against gradient based methods, such as the Fast Gradient Sign Method \cite{goodfellow2014explaining} and Ensemble Adversarial Training \cite{DBLP:conf/iclr/TramerKPGBM18}. It is interesting to investigate why tropical neural networks are robust against such attacks.  In addition, it is interesting to develop adversarial attacks toward tropical neural networks.

\bibliographystyle{plain} 
\bibliography{refs}

\end{document}